\documentclass[12pt]{article}
\usepackage[round]{natbib}
\usepackage{epsfig}
\usepackage{relsize, amsmath,amsfonts}

\usepackage{amsfonts}
\usepackage{url}
\usepackage{appendix}
\usepackage{verbatim}
\usepackage{setspace}
\usepackage{graphicx}
\newcommand{\blind}{1}

\usepackage[top=1in,bottom=1in,left=1in,right=1in]{geometry}
\newtheorem{theorem}{Theorem}

\newcommand{\qed}{\nobreak \ifvmode \relax \else
      \ifdim\lastskip<1.5em \hskip-\lastskip
      \hskip1.5em plus0em minus0.5em \fi \nobreak
      \vrule height0.75em width0.5em depth0.25em\fi}

\usepackage{amssymb} %mathbb
\newcommand{\abs}[1]{\left\lvert#1\right\rvert} %abs

\def\T{{ \mathrm{\scriptscriptstyle T} }}

\def\T{{ \mathrm{\scriptscriptstyle T} }}

\newcommand{\norm}[1]{\left\Vert#1\right\Vert}

\begin{document}

\date{}

\if1\blind
{
  \title{\LARGE\bf Bayesian Variable Selection for Skewed Heteroscedastic Response}
\author{
\small
\begin{tabular}{c}
Libo Wang$^1$,
Yuanyuan Tang$^1$,
Debajyoti Sinha$^1$,
Debdeep Pati$^1$,
and Stuart Lipsitz$^2$\\
$^1$ Department of Statistics, Florida State University, Tallahassee, FL 32306    \\
$^2$ Brigham and Women's Hospital \\
\end{tabular}
\normalsize}
  \maketitle
} \fi

\if0\blind
{
  \bigskip
  \bigskip
  \bigskip
  \begin{center}
    {\LARGE\bf Bayesian Variable Selection for Skewed Heteroscedastic Response}
\end{center}
  \medskip
} \fi

%\begin{spacing}{1.9}

\bigskip
\begin{abstract}
In this article, we propose new Bayesian methods with proper theoretical justification for selecting and estimating a sparse regression coefficient vector for skewed heteroscedastic response. Our novel Bayesian procedures effectively estimate the median and other quantile functions, accommodate non-local prior for regression effects without compromising ease of implementation
via sampling based tools. First time for skewed and heteroscedastic response, this Bayesian method asymptotically selects the true set of predictors even when the number of covariates increases in the same order of the sample size. We also
extend our method to deal with some observations with very large errors. Via simulation studies and a re-analysis of a medical cost study with large number of potential predictors, we illustrate the ease of implementation and other practical advantages of our approach compared to existing methods for such studies.
\end{abstract}

\noindent%
\vfill

\section{Introduction} \label{sec:intro}
Large number of possible predictors and highly skewed heteroscedastic response are often major challenges for
many biomedical and econometric applications.  Selection of an optimal set of covariates and subsequent estimation of the regression function are important steps for scientific conclusions and policy decisions based on such studies. For example, previous analyses of Medical Expenditure Panel Study \citep{natarajan2008variance, cohen2003design} testify to the highly skewed and heteroscedastic nature of the main response of interest, total health care expenditure in a year. Also, it is common in such studies to have a small proportion of patients with either very high or very low medical costs.  Popular classical sparse-regression methods such as Lasso (Least absolute shrinkage operator) by \citet{tibshirani1996regression} and \citet{efron2004least}, and later related methods of \citet{fan2001SCAD}, \citet{Zou05regularizationand}, \citet{zou2006adap} and MCP \citep{zhang2010nearly} assume Gaussian (or, at least symmetric) response density with common variance. Limited recent literature on consistent variable selection for non-Gaussian response includes \citet{Zhao:2006:MSC:1248547.1248637} under common variance assumption, \citet{bach2008} under weak conditions on covariate structures, and \citet{Chen2014125} under skew-t errors. However, none of these methods deal with estimation of quantile function for heteroscedastic response frequently encountered in complex biomedical  studies. Many authors including \citet{koenker2005quantile} argue effectively against focusing on mean regression for skewed heteroscedastic response. Our simulation studies demonstrate that directly modeling  skewness and heteroscedasticity, particularly in presence of analogous empirical evidence, leads to better estimators of quantile functions for finite samples compared to existing methods which ignore skewness and heteroscedasticity.

Bayesian methods for variable selection have some important practical advantages including incorporation of prior information about sparsity, evaluation of uncertainty about the final model, interval estimate for any coefficient of interest, and evaluation of the relative importance of different coefficients. Asymptotic properties of Bayesian variable selection methods when the number of potential predictors, $p$, increases as a function of sample size $n$ have received lot of attention recently in the literature. Traditionally, to select the important variables out of $(X_1, \ldots, X_p)$,  a two component mixture prior, also referred to as ``spike and slab'' prior,  \citep{mitchell1988bayesian,george1993variable,george1997approaches} is placed on the coefficients $\beta=(\beta_1,\ldots,\beta_p)$. These mixture priors include a discrete mass, called a ``spike'', at zero to characterize the prior probability of a coefficient being exactly zero (that is, not including the corresponding predictor in the model) and a continuous density called a ``slab'', usually centered at null-value zero, representing the prior opinion when the coefficient is non-zero. Following \citet{johnson2010use,johnson2012nonlocal}, when the continuous density of the slab
part of a spike and slab prior has value 0 at null-value 0, we will
call it a non-local mixture prior. Continuous analogues of local mixture priors are being proposed recently by \citet{park2008bayesian,carvalho2010horseshoe,bhattacharya2014dirichlet} among others.  \cite{bondell2012consistent} presented the  selection consistency of joint Bayesian credible sets. However, current Bayesian variable selection methods usually focus on mean regression function for models with symmetric error density and common variance.

\cite{johnson2012nonlocal} recently showed a startling selection inconsistency phenomenon for using several commonly used mixture priors, including local mixture (spike and slab prior
with non-zero value at null-value 0 of the slab density) priors, when $p$ is larger than the order of $\sqrt{n}$. To address this for mean regression with sparse $\beta$, they advocated the use of non-local mixture density presenting continuous ``slab"  density with value 0 at null-value 0 because these priors, called non-local mixture priors here, obtain selection consistency when the dimension $p$ is $O(n)$. \cite{2014arXiv1403.0735C} provided several conditions  to ensure selection consistency even when $p \gg n$. However, none of these Bayesian methods specifically deal with skewed and heteroscedastic response, contamination of few observations with
 large errors and variable selection for median and other quantile functionals.

In this article, we accommodate skewed and heteroscedastic response distribution using transform-both-sides model \citep{lin2012semiparametric} with sparsity inducing prior for the vector of regression coefficients.  Our key observation is that, under such models with generalized Box-Cox transformation \citep{bickel1981analysis}, even a local mixture prior on after-transform regression coefficients induces non-local priors on the original regression function for certain choices of the transformation parameter. Similar to moment and inverse moment non-local priors in \citet{johnson2012nonlocal}, this method enables clear demarcation between the signal and the noise coefficients in the posterior leading to consistent posterior selection even when $p = O(n)$. Addition to that, our use of standard local priors on the transformed regression coefficients facilitates  straightforward posterior computation which can be implemented in publicly available softwares.  We later extend this model to accommodate 
cases when the observations are contaminated with a small number of observations with very large (or small) 
errors. Our approaches are shown to out-perform well-known competitors in simulation studies as well as for analyzing and interpreting a real-life medical cost study.

\section{Bayesian variable selection model}
\subsection{Transform-both-Sides Model}
For the skewed and heteroscedastic response $Y_i$ for $i=1,\ldots,n$, we assume the transform-both-sides model \citep{lin2012semiparametric}
\begin{equation} \label{eq:transmodel}
g_\eta(Y_i)=g_\eta(x_i^{\mathsmaller T}\beta)+e_i\ ,
\end{equation}
where $\beta=(\beta_1,\ldots,\beta_p)'$, $x_i$ is the observed $p$-dimensional covariate vector, $g_{\eta}(y)$ is the monotone power transformation \citep{bickel1981analysis}, %an extension of the Box-Cox power family,
\begin{equation}\label{eq:boxcox}
g_\eta(y)=\frac{y|y|^{\eta-1}-1}{\eta},
\end{equation}
with unknown parameter $\eta \in (0,2)$. This transformation in \eqref{eq:boxcox} is an extension of Box-Cox power family that has a long history and success in dealing with skewed and heteroscedastic response. We assume that under an optimal $\eta$, the transformed response $g_{\eta}(y)$ has a symmetric and unimodal distribution with mean and median $g_{\eta}(x_i^{\mathsmaller T}\beta)$. Thus $e_i$'s are independent mean 0 errors with common symmetric density function $f_e$ and variance $\sigma^2$.
%and $e_i$'s are independent mean 0 errors with common symmetric density function $f_e$ and variance $\sigma^2$. 
The transformation $g_\eta(y)$ in \eqref{eq:boxcox} is monotone with derivative $g_\eta'(y)=|y|^{\eta-1} \geq 0$. Model \eqref{eq:transmodel} can be expressed as a linear model
\begin{equation} \label{eq:medianmodel}
Y_i\ =\ x_i^{\mathsmaller T}\beta\ +\ \epsilon_i\ ,
\end{equation}
where $\epsilon_i$ has a skewed heteroscedastic density with median 0 because $P[\epsilon_i>0]=1/2$, and approximate variance is $\sigma^2|x_i^{\mathsmaller T}\beta|^{2-2\eta}$. Hence, the median of the skewed and heteroscedastic response $Y_i$ in \eqref{eq:transmodel} is $x_i^{\mathsmaller T}\beta$.  For the time being, we consider a Gaussian $\mbox{N}(0, \sigma^2)$ density for $f_e$ in \eqref{eq:transmodel}. Later in \S \ref{sec:gen}, we consider other densities to accommodate a heavy tail for $f_e$.

%As discussed  in \S \ref{sec:intro}, independent spike and slab priors for every $\beta_j$ constitute a popular choice to induce sparsity in $\beta$.  
For the model of \eqref{eq:transmodel}, any sparsity inducing prior for $\beta$  should depend on the transformation parameter $\eta$ since $\eta$ has a significant effect on the range and scale of $Y_i$ (approximate variance $\sigma^2|x_i^{\mathsmaller T}\beta|^{2-2\eta}$). Based on this argument, we specify an conditional mixture prior $g_{\eta}(\beta_j)$ given $\eta$ using a ``local'' $\phi(\cdot ; 0, \sigma_\beta^2)$ density
for the ``slab'' when $g_{\eta}(\beta_j)$ is non-zero with discrete prior probability $(1-\pi_0)$, where  $\phi(\cdot ; \mu, v^2)$ is the Gaussian density with mean $\mu$ and variance $v^2$. This conditional
mixture prior $g_{\eta}(\beta_j)$ given $\eta$ (a local mixture prior according to definition of \cite{johnson2012nonlocal}) results in a conditional mixture prior
\begin{eqnarray}\label{eq:betaprior}
	f_{\beta}(\beta_j \mid \eta)=\pi_0\delta_0(\beta_j)+(1-\pi_0)\phi(g_\eta(\beta_j); 0,\sigma_\beta^2)|g_\eta'(\beta_j)|
\end{eqnarray}
for $\beta_j$ given $\eta$ independently for $j=1,\cdots,p$, where $\pi_0 \in [0, 1]$ is the probability of $\beta_j$ being zero, $\delta_0(\cdot)$ is the discrete measure at 0. When $\eta >1$,   $g_\eta'(0) = 0$ and hence the resulting unconditional marginal prior for the nonzero $\beta_j$ in \eqref{eq:betaprior} turns out to be a non-local mixture prior of \cite{johnson2012nonlocal}. However, the prior of transformed $g_\eta(\beta_j)$, a mixture of discrete measure at 0 and $\phi(\cdot; 0, \sigma_{\beta}^2)$ density, is a local mixture prior. This is demonstrated via the plots of two resulting unconditional priors of $\beta_j$ when $\eta =0.5$ and $\eta=1.5$ in Figure \ref{fig:density_eta}. Our simulation study in \S \ref{sec:sim} shows that the model selection and estimation procedures for our Bayesian method perform substantially better than competing methods when $\eta>1$ (case with non-local unconditional prior for $\beta_j$) compared to, say, when $\eta=0.5 \in (0,1)$ (case with a local prior for $\beta$  given $\eta$).  Thus, heteroscedasticity and the possibly non-local property of $\pi(\beta \mid \eta)$  come as a bi-product of the transform-both-sides model of \eqref{eq:transmodel}. This implicit non-local mixture prior modeling of unconditional $\beta_j$ may be the reason for some desirable asymptotic properties of our method even when $p=O(n)$ (discussed in \S \ref{sec:cons}). However, our methods' ability to use a local prior for $g_{\eta}(\beta_j)$ significantly reduces computational complexity of the associated Markov chain Monte Carlo (MCMC) algorithms, while facilitating the desirable asymptotic property.

% Even though our possibly non-local prior of $\pi(\beta \mid \eta)$ in \eqref{eq:betaprior} enjoys the desirable theoretical properties (as mentioned in \citet{johnson2012nonlocal}), the prior is convenient to implement in MCMC computation because it is equivalent to using independent local mixture prior on transformed coefficient $g_{\eta}(\beta_j)$.

%%%%%%Fig 1
\begin{figure}[htp!]
\centering
\includegraphics[scale=0.25]{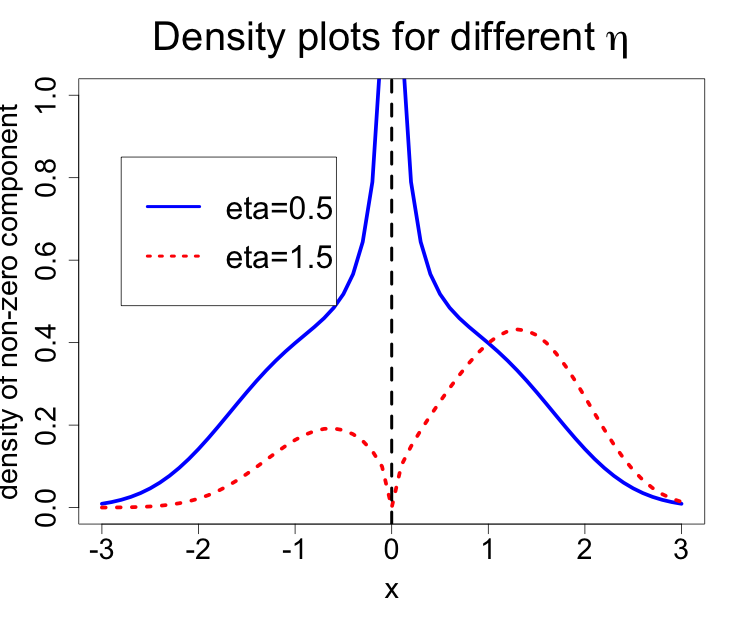}
\caption{Density Plot for Different $\eta$}
\label{fig:density_eta}
\end{figure}

When the error density $f_e$ in \eqref{eq:transmodel} is Gaussian $\phi(\cdot ; \mu, \sigma^2)$ with mean 0 and variance $\sigma^2$, we can specify a hierarchical Bayesian model using the prior $\pi(\beta \mid \eta)$ in \eqref{eq:transmodel} along with known priors for $\sigma^2$ and $\eta$,
\begin{eqnarray}  \label{eq:basicBayes}	
	\sigma^2 \sim \mbox{IGa}(a,b),	\quad
	\eta/2 \sim  \mbox{Beta}(c_1,d_1) \label{eq:prior}
\end{eqnarray}
and a hyperprior for the hyperparameters $(\pi_0,\sigma_\beta^2)$. For computational simplicity, in \S \ref{sec:cons}, we establish variable selection consistency of this hierarchical Bayesian model using \eqref{eq:transmodel} along with the non-local prior \eqref{eq:betaprior} on $\beta$ and the prior for $\sigma^2$ in \eqref{eq:basicBayes}.

\section{Consistent variable selection for large $p$}\label{sec:cons}
In this section, we investigate selection consistency of the proposed model when the number of covariates $p$ is grows with sample size $n$, with $p \leq n$ and the true error density is skewed and heteroscedastic, but follows the same specification as  in \eqref{eq:transmodel}.  Unlike the Gaussian likelihood with mixture priors of \citet{johnson2012nonlocal}, our Bayesian model described in \eqref{eq:transmodel} - \eqref{eq:basicBayes} does not admit a closed form expression of the marginal likelihood. We derive appropriate bounds  of the marginal likelihood  to obtain the desired Bayesian consistency results. Here we only present a brief outline of our assumptions, developments and practical implications of our theoretical results. Supporting results and details of the proofs are deferred to \S \ref{ssec:A1}. For brevity of exposition, we  consider a design matrix $X$ which are nearly orthogonal in the sense that there exists constant $0 < c_1 \leq c_1$ such that  $c_1 n \leq \lambda_1(X^{\mathsmaller T} X) \leq \lambda_p(X^{\mathsmaller T} X) \leq c_2 n$ where $\lambda_1(A) \leq \lambda_2(A) \leq \cdots \leq \lambda_p(A)$ denote ordered eigen values of the matrix $A$.  The assumption ensures identifiability of the regression coefficient $\beta$ and is commonly assumed in the selection consistency literature, refer for example to  \cite{johnson2012nonlocal}.  % In addition, we adopt a further identifiability restriction on the design matrix.   We assume that any two coefficient vectors $\beta_1 \neq \beta_2$, there exists a positive number 
%$\gamma$ 
%\begin{equation}\label{eq:identifiable}
%\sum_{i=1}^{n}(g_{\eta}(x_i^T\beta_1)-g_{\eta}(x_i^T\beta_2))^2 \geq \gamma
%\end{equation}
%  Note that if the columns of the design matrix $X$ are linearly dependent \eqref{eq:identifiable} is not satisfied.  \textcolor{red}{However the assumption is stronger than assuming linear independence since we require the $\gamma$ to be a universal constant non depending on $\beta_1$ and $\beta_2$.}

To estimate the posterior probability assigned to the correct model, we use Laplace approximation \citep{rossell2015non} to obtain probabilistic bounds of marginal likelihood $p(y|\gamma)$. Here $\gamma =(\gamma_1, \ldots, \gamma_p)$ denotes the predicted model for which  $\gamma_j=I(\beta_j \neq 0)$ for $j= 1, 2, \cdots, p$ are the indicators of active coefficients and $p_\gamma=\sum_{i=1}^p\gamma_i$ denotes the number of active variables. Denote by $\beta^*$ the vector of true regression coefficients with $\gamma^*$ defined accordingly. Assuming $\tilde{\beta}$ as the posterior mode and $\eta$ is the optimal power parameter, the Laplace approximation of $p(y|\gamma)$ around $\tilde{\beta}$ is 
\begin{equation}\label{eq:marginlikelihood}
p(y|\gamma) \approx \frac{(2\pi)^{p_\gamma/2}}{|H(\tilde{\beta})|^{1/2}}
\exp\{\log L( \tilde{\beta})+\text{log}(\pi(\tilde{\beta}|\gamma)))\}
\end{equation}
where the projection of the prior $f_\beta$ in \eqref{eq:betaprior} onto the support of $\gamma$, denote by $\pi(\beta \mid \gamma)$  and is given as 
\begin{eqnarray}\label{eq:nonlocal}
	\pi(\beta \mid \gamma) =  \prod_{j \in \gamma}  \phi(g_\eta(\beta_j); 0,\sigma_\beta^2)\abs{g_\eta'(\beta_j)}.
\end{eqnarray}
We note again that \eqref{eq:nonlocal} becomes a non-local prior \citep{johnson2012nonlocal} for $\eta \in (1,2)$ and is the primary reason for controlling false positives.  Based on definitions of key concepts in \eqref{eq:marginlikelihood} and \eqref{eq:nonlocal}, we state our main theorem on selection consistency of our Bayesian method even when $p$ is of the order $O(n)$.
 \begin{theorem}\label{thm:1}
 When the observations are generated from \eqref{eq:transmodel} for a known $\eta \in (1, 2)$ and $\sigma > 0$, and model \eqref{eq:transmodel} is fitted with priors \eqref{eq:betaprior} and \eqref{eq:prior} for the optimal $\eta$ and with any $\pi_0 \in (0, 1)$, then for $p \leq n$,   the posterior probability $P( \gamma=\gamma^* \mid y)\to 1$ almost surely as $n, p \to+\infty$ .
\end{theorem}
The detailed proof of Theorem 1, given in \S \ref{ssec:A1}, is a non-trivial extension of the proof of Theorem 1 in \cite{johnson2012nonlocal} which used non-local priors to obtain variable selection consistency. Unlike them, we use a local prior $g_{\eta}(\beta)$ to induce a possible non-local prior for $\beta$ given $\eta$ in \eqref{eq:betaprior}. To the best of our knowledge, this is the first result on  Bayesian selection consistency when the response distribution is skewed and heteroscedastic and our method of proof opens the theoretical investigation of sparse Bayesian methods for transformable models and heterodasticity.

\section{Accommodating extremely large errors}\label{sec:gen}
Presence of few observations with extremely large errors and their influences on final analysis for various application areas have been emphasized by many authors including \citet{hampel2011robust}.  The assumption of Gaussian error density $f_e$ in \eqref{eq:transmodel} may not be valid due to the presence of a small number of observations with large errors even after optimal Box-Cox transformation. To address this,  we extend the model \eqref{eq:transmodel} to a random location-shift model with
\begin{equation}\label{eq:transmodeloutlier}
g_\eta(Y_i)=g_\eta(x_i^{\mathsmaller T}\beta)+\gamma_i+e_i\ ,
\end{equation}
where $\gamma_i$ is nonzero if the $i$th observation has large error, and  zero otherwise. We assume the vector $\gamma = (\gamma_1, \ldots, \gamma_n)^{\mathsmaller T}$ to be sparse to ensure only a small probability of the response having a large error after transformation. Similar idea of location-shift model, however, with un-transformed response, is popular in the recent literature on robust linear models (for example, \citet{she2011outlier} and \citet{mccann2007robust}).  To ensure that $g_\eta(x_i^{\mathsmaller T}\beta)$ is the mean and median of $g_\eta(Y_i)$,  we require the mean and median of $\gamma_i+e_i$ to be zero, that is, we need a symmetric distribution for $\gamma_i$. For this purpose, we use another spike-and-slab mixture prior $f_{\gamma}(\gamma_i)=\pi_\gamma\delta_0+(1-\pi_\gamma)\phi(\gamma_i; 0,\sigma_\gamma^2)$ independently for $i=1, \ldots, n$, where $0<\pi_{\gamma}<1$.

To induce a heavy-tailed error density after transformation, we also consider another extension of the model \eqref{eq:transmodel} as
\begin{eqnarray}\label{eq:errorNI}
g_\eta(Y_i)=g_\eta(x_i^{\mathsmaller T}\beta)+U_i^{-1/2}e_i\qquad with \quad U_i \sim H(\cdot \mid \nu)\ ,
\end{eqnarray}
where $H(\cdot \mid \nu)$ is a positive mixing distribution indexed by a parameter $\nu$ and $e_i$'s are again independent $\mbox{N}(0,\sigma^2)$.  This class of heavy-tailed error distributions of \eqref{eq:errorNI} is called normal independent (NI) family \citep{lange1993normal}. We consider three kinds of heavy tailed distribution, Student's-t, slash and contaminated normal (CN) respectively,
for the marginal error density in \eqref{eq:errorNI} using the following specific choices of $H(\cdot \mid \nu)$ \citep{lachos2011linear}: $\chi_\nu^2/\nu$ distribution with possibly non-integer $\nu>2$, $H(u\mid \nu)=u^{\nu}$ for $u\in [0,1]$, and discrete $H(u\mid \nu)$ with $P[\rho<1]=1-P[\rho=1]=\nu$. For student-t error, we use the prior for the degrees of freedom parameter $\nu$ to be a truncated exponential on the interval $(2,\infty)$.  For $\nu$ of the slash distribution marginal error, we use a $\mbox{Gamma}(a,b)$ prior with small positive values of $a$ and $b$ with $b\ll a$. For contaminated normal marginal error, we assign $\mbox{Beta}(\nu_0,\nu_1)$ and $\mbox{Beta}(\rho_0,\rho_1)$ priors respectively for $\nu$ and $\rho$. In \S \ref{sec:sim}, we compare the performances of Bayesian analyses under these competing models for highly skewed and heteroscedastic responses.

\section{Simulation Studies}\label{sec:sim}
{\bf Simulation model with no outliers:} We use different simulation models to compare our Bayesian methods under model \eqref{eq:transmodel} with LASSO \citep{tibshirani1996regression} and the penalized quantile methods \citep{koenker2005quantile}. From each simulation model, we simulated $50$ replicated datasets of sample size $n=50$. For both simulation studies, the observations are sampled from the model \eqref{eq:transmodel} with $e_i \sim \mbox{N}(0,\sigma_0^{2})$ with $\sigma_0^{2}=1$. The hyperparameters for priors in \eqref{eq:basicBayes} are set as $a=2, b=2, c_1=1$ and $d_1=1$.
%For the penalized quantile regression, the  output from the code does not provide exact zeros. Hence we use a threshold 0.01 to distinguish exact zero estimated coefficients from non-zero estimated coefficients.
The tuning parameters for LASSO and penalized quantile regression are selected via a grid search based on the 5-fold cross-validation. We compare the estimators from different methods based on following criteria: the mean masking proportion $M$ (fraction of undetected true $\beta_j \neq 0$), the mean swamping proportion $S$ (fraction of wrongly selected $\beta_j=0$), and the joint detection rate $\mbox{JD}$ (fraction of simulations with 0 masking). We also compare the goodness-of-fit of  estimation methods using an influence measure $L/L^*-1$, where 
\begin{align} \label{eq:estimationL1}
L&=&\sum_{i=1}^n(g_{\eta_0}(y_i)-g_{\eta_0}(x_i^{\mathsmaller T}\hat{\beta}))^2/(2\sigma_0^{2})-n/2\log(2\pi\sigma_0^{2})+(\eta_0-1)\sum_{i=1}^n\log(|y_i|).
\end{align}
is the log-likelihood under $(\hat{\beta},\eta_0,\sigma_0)$ and $L^*$ is the same log-likelihood under $(\beta_0, \eta_0,\sigma_0)$, and  $(\beta_0,\eta_0,\sigma_0)$ are the known true parameter values (of the simulation model).  The results of our study using simulated data from TBS model \eqref{eq:transmodel} with different values of $\eta$ are displayed in Table \ref{table:snonoutlier} with $p=8$.

%%%%%Table 1
\begin{table}[htp!]
\begin{center}
\caption{Results of simulation studies for using different methods of analysis: \label{table:snonoutlier}}
\smallskip
\small{Simulation model of \eqref{eq:transmodel} with $\eta_0=0\mbox{$\cdot$}5, p=8, \beta_0=(3,1\mbox{$\cdot$}5,0,0,2,0,0,0)$.}
\scalebox{0.9}{
\begin{tabular}[ht!]{l|llllll}
	\hline 
	Method used & $L/L^*-1$& \# of non-zeros &M($\%$) &S($\%$) &\mbox{JD}($\%$)\\ \hline
	TBS-SG& -0$\cdot$02 & 3$\cdot$16 & 0 & 3$\cdot$2 & 100\\ 
	Penalized Quantile&	0$\cdot$04 & 5$\cdot$84 & 0 & 56$\cdot$8 & 100\\ 
	LASSO & 0$\cdot$04 & 4$\cdot$98 & 0 & 3$\cdot$96 & 100\\ 
	TBSt-SG&	-0$\cdot$01 & 3$\cdot$16 & 0 & 3$\cdot$2&100	\\ 
	TBSS-SG &-0$\cdot$02 & 3$\cdot$16 & 0 & 3$\cdot$2 &100		\\	
	TBSCN-SG & -0$\cdot$02 & 3$\cdot$14 & 0 & 2$\cdot$8 & 100 	\\	\hline
\end{tabular}}
\end{center}

\begin{center}
\small{Simulation model of \eqref{eq:transmodel} with $\eta_0=1\mbox{$\cdot$}8, p=8, \beta_0=(3,1\mbox{$\cdot$}5,0,0,2,0,0,0)$.}
\scalebox{0.8}{
\begin{tabular}[ht!]{l|llllll}
	\hline 
	Method used & $L/L^*-1$& \# of non-zeros &M($\%$) &S($\%$) &\mbox{JD}($\%$)\\ \hline
	TBS-SG& -0$\cdot$06 & 3$\cdot$02 &0 & 0$\cdot$4&100\\ 
	Penalized Quantile&	0$\cdot$04 & 5$\cdot$82 & 0 & 56$\cdot$4 & 100\\
	LASSO & 0$\cdot$66 & 4$\cdot$56 &0 & 3$\cdot$12&100\\ 
	TBSt-SG & -0$\cdot$05 & 3 & 0 & 0 &100		\\ 
	TBSS-SG & -0$\cdot$06 & 3 & 0 & 0 & 100		\\	
	TBSCN-SG &-0$\cdot$06 & 3 & 0 & 0 &100 	\\	\hline
\end{tabular}}
\smallskip

\small{$M$: masking proportion (fraction of undetected true $\beta_j \neq 0$); $S$: swamping proportion $S$ (fraction of wrongly selected $\beta_j$ with true value 0); $\mbox{JD}$: joint detection rate.}
\end{center}
\end{table}

In Table \ref{table:snonoutlier}, we compare our Bayesian TBS model \eqref{eq:transmodel} with prior \eqref{eq:betaprior} for $\beta$ (called TBS-SG in short) to frequentist methods of penalized quantile and LASSO. From the results in Table \ref{table:snonoutlier},  it is evident that our TBS-SG method provides better results than competing frequentist methods in terms of average number of non-zeros and swamping error rate. We also compare TBS-SG method with other Bayesian TBS models with heavy tailed normal independence (NI) error in \eqref{eq:transmodel}. These competing NI models in \eqref{eq:errorNI} include TBSt-SG model (in short) with $t$ distribution for $H(\cdot \mid \nu)$, TBSS-SG model (in short) with slash distribution for $H(\cdot \mid \nu)$ and TBSCN-SG model (in short) with contaminated normal distribution for $H(\cdot \mid \nu)$. All our Bayesian methods have ``SG'' in their end of acronym to indicate the spike Gaussian prior of \eqref{eq:betaprior} for $\beta$. TBS models accommodating heavy tailed response perform the best in competing models with ideal masking, swamping and joint outlier detection rates. Both our methods and frequentist methods provide comparable performances based on average $L/L^*-1$ values, although the $L/L^*-1$ values from penalized quantile estimates using different datasets are highly variable. All methods perform desirable with respect to masking and joint detection. Also,  we found that our Bayesian methods provide better results when true $\eta$ value is $\eta_0=1.8$, compared to $\eta_0=0.5$.

To compare the performances for different $\eta$, we set $p =20$ and the number of non-zero coefficient to be $12$.  Denote by $(x)_{k}$ the vector formed by appending $k$ copies of $x$. Consider case i) $\beta_0= \{(2)_{12}, (0)_8\}$,  case ii) $\beta_0=\{(-10)_6, (4)_6, (0)_8 \}$, case iii) $\beta_0=\{(-10)_{10}, (4)_2, (0)_8 \}$, case iv) $\beta_0= \{(-10)_2, (-4)_2, (-2)_2, (2)_2, (4)_2, (10)_2, (0)_8  \}$, case v) $\beta_0= \{(-10)_6, (2)_6, (0)_8\}$, case vi) $\beta_0= \{(-10)_2, (-8)_2, (-6)_2, (-4)_2, (-2)_2, (2)_2, (0)_8 \}$.  We use only TBSCN-SG model for analysis because these three TBS models accommodating heavy tailed response have similar performance.

%%%%%Table 2
\begin{table}[htp!]
	\begin{center}
	\caption{Results of simulation studies for using different methods of analysis when the true model is \eqref{eq:transmodel} with $p=20$: \label{table:sgaussian}}
	\smallskip \smallskip
	\scalebox{0.85}{
		\begin{tabular}{ll|ll|ll|ll|ll}
			\hline
			\multicolumn{2}{l}{} & \multicolumn{2}{|l}{TBS-SG} & \multicolumn{2}{|l}{Penalized Quantile} & \multicolumn{2}{|l}{LASSO} & \multicolumn{2}{|l}{TBSCN-SG} \\
			 \hline
			&  Measurement & $\eta_0$=0$\cdot$5 & $\eta_0$=1$\cdot$8 & $\eta_0$=0$\cdot$5 & $\eta_0$=1$\cdot$8 & $\eta_0$=0$\cdot$5 & $\eta_0$=1$\cdot$8& $\eta_0$=0$\cdot$5 & $\eta_0$=1$\cdot$8 \\ \hline
			\hline
			Case i)& $L/L^*-1$&-0$\cdot$05 &2$\cdot$72  &0$\cdot$02  &0$\cdot$7  & 0$\cdot$07 & -82$\cdot$81 & -0$\cdot$04 & 2$\cdot$74 \\ \cline{2-10}
			&\# of non-zeros&13$\cdot$6&12$\cdot$02&16$\cdot$36&15$\cdot$58 & 14$\cdot$02 &12$\cdot$24 & 13$\cdot$74 & 12$\cdot$02\\ \cline{2-10}
			&M ($\%$)& 0 & 0 & 0$\cdot$5 & 0 & 0 & 0 &0 &0\\ \cline{2-10}
			&S($\%$)& 20 & 0$\cdot$25 & 55$\cdot$25 & 44$\cdot$75 & 25$\cdot$75 &3 & 21$\cdot$75 & 0$\cdot$25 \\ \cline{2-10}
			&\mbox{JD}($\%$)&100&100&94&100 & 96 & 100 &100 &100\\ \hline \hline
			Case ii)&$L/L^*-1$ & -0$\cdot$03 & 0$\cdot$03  & 0$\cdot$05  & -0$\cdot$24 & 0$\cdot$1  &-1282 & -0$\cdot$02 & 0$\cdot$09 \\ \cline{2-10}
			& \# of non-zeros & 13$\cdot$26 & 12 & 17$\cdot$12 & 15$\cdot$02&14$\cdot$96&12$\cdot$14 & 13$\cdot$24 & 12 \\ \cline{2-10}
			&M ($\%$)& 0$\cdot$67 & 0 & 0$\cdot$67 & 0&0&0&8$\cdot$33&0\\ \cline{2-10}
			&S($\%$)&16$\cdot$75 & 0 & 65 & 37$\cdot$75&37$\cdot$5&1$\cdot$75& 16$\cdot$75 & 0\\ \cline{2-10}
			&\mbox{JD}($\%$)&94&100&94&100 &96 &100&92&100\\ \hline \hline
			Case iii)&$L/L^*-1$&-0$\cdot$03 & -0$\cdot$06  & 0$\cdot$03 & -0$\cdot$15 & 0$\cdot$11 &-2736 & -0$\cdot$02& 0$\cdot$04\\ \cline{2-10}
			& \# of non-zeros& 12$\cdot$94 & 12 & 16$\cdot$94 & 14$\cdot$52 &14$\cdot$7 &12$\cdot$24& 12$\cdot$92& 12\\ \cline{2-10}
			&M ($\%$)&0$\cdot$33 & 0 & 0$\cdot$33 & 0 & 0$\cdot$83 &0 &0$\cdot$5&0\\ \cline{2-10}
			&S($\%$)&12$\cdot$25 & 0 & 62$\cdot$25 & 31$\cdot$5 & 35 &3& 12$\cdot$25& 0\\ \cline{2-10}
			&\mbox{JD}($\%$)&96&100&96&100 &90 &100&94&100\\ \hline \hline
			Case iv)&$L/L^*-1$&-0$\cdot$01  & 0$\cdot$09  & 0$\cdot$05  & -0$\cdot$22   & 0$\cdot$09 & -689$\cdot$6& -0$\cdot$01& 0$\cdot$13\\ \cline{2-10}
			& \# of non-zeros&12$\cdot$58 & 12 & 16$\cdot$42 & 15$\cdot$32 & 14$\cdot$2& 12$\cdot$08& 12$\cdot$18& 12\\ \cline{2-10}
			&M ($\%$)&6$\cdot$17& 0 & 3$\cdot$67 & 0 &5$\cdot$67 & 0$\cdot$1&7$\cdot$67&0\\ \cline{2-10}
			&S($\%$)&16$\cdot$5 & 0 & 60$\cdot$75 & 41$\cdot$5 &36 &2& 13$\cdot$75& 0 \\ \cline{2-10}
			&\mbox{JD}($\%$)&46&100&64&100 &54 &92&34&100\\ \hline \hline
			Case v)&$L/L^*-1$&0$\cdot$01  & 0$\cdot$04  & 0$\cdot$05  & -0$\cdot$21& 0$\cdot$1  &-1072& 0$\cdot$01& 0$\cdot$09\\ \cline{2-10}
			&\# of non-zeros&11$\cdot$32 & 12 & 16$\cdot$12 & 14$\cdot$96 & 13$\cdot$56 & 12$\cdot$04& 11$\cdot$16& 12\\ \cline{2-10}
			&M ($\%$)&11$\cdot$17 & 0 & 5 & 0 &7$\cdot$17 & 0$\cdot$83&12&0\\ \cline{2-10}
			&S($\%$)&8$\cdot$25 & 0 & 59 & 37 &30$\cdot$25 &  1$\cdot$75& 7$\cdot$5& 0\\ \cline{2-10}
			&\mbox{JD}($\%$)&24&100 &62&100 &34 &90& 14& 100\\ \hline \hline
			Case vi)&$L/L^*-1$&-0$\cdot$01  & 0$\cdot$10  & 0$\cdot$06  & -0$\cdot$27 & 0$\cdot$12 & -763$\cdot$3& 0$\cdot$01& 0$\cdot$13 \\ \cline{2-10}
			&\# of non-zeros&12$\cdot$1 & 12 & 15$\cdot$62 & 14$\cdot$96 & 13$\cdot$52 & 12$\cdot$14& 11$\cdot$76& 12\\ \cline{2-10}
			&M ($\%$)&6$\cdot$67 & 0 & 4$\cdot$83 & 0 & 8$\cdot$33 & 1$\cdot$33&8$\cdot$33&0\\ \cline{2-10}
			&S($\%$)&11$\cdot$25 & 0 & 52$\cdot$5 & 37& 31$\cdot$5 & 3$\cdot$75& 9$\cdot$5& 0\\ \cline{2-10}
			&\mbox{JD}($\%$)&44&100 &62&100 &36 &84&36&100\\ \hline
		\end{tabular}}
		
$M$: masking proportion; $S$: swamping proportion $S$; $\mbox{JD}$: joint detection rate.
\end{center}
\end{table}

From the results in Table \ref{table:sgaussian}, we can clearly see that for all the cases, all the four methods perform better when $\eta_0=1.8$ compared to when $\eta_0=0\mbox{$\cdot$}5$, with respect to average number of non-zeros, masking, swamping and joint detection rate. This can be explained by the fact that when $\eta_0=1\mbox{$\cdot$}8$, we expect the posterior draws of $\eta$ to be close to $1\mbox{$\cdot$}8$ which corresponds to a non-local prior for $\beta$ (see Figure \ref{fig:density_eta}). When the range of signals is large and when there are many groups of small coefficients (see case (iv) and case (vi)), all of the methods do not perform well. Considering only variable selection results (average number of non-zeros), our TBS model clearly out performs penalized quantile method and LASSO.\\

\noindent {\bf Studies using simulation model with outliers and heavy-tailed distribution:} Our simulation models are similar to previous simulation model of \eqref{eq:transmodel} except that a few of the observations are now have large errors even after transformation. Although the Bayesian TBS methods with NI error in \eqref{eq:errorNI} do not provide the identification and estimation of these observations, we wonder whether they ensure robust variable selection and estimation of $\beta$, particularly in comparison to the Bayesian method using random location-shift model of \eqref{eq:transmodeloutlier}.

For the sake of brevity of the presentation, we omit the tables for results of simulation studies using data simulated from models \eqref{eq:transmodeloutlier} and \eqref{eq:errorNI}, and only summarize the results here. When we use the simulation model \eqref{eq:transmodeloutlier} with $\eta_0=0\mbox{$\cdot$}5, p=8, \gamma_{(1:2)}=8, \gamma_3=-8, \gamma_{(4:50)}=0$, and $\beta_0=(3,1\mbox{$\cdot$}5,0,0,2,0,0,0)$, our Bayesian method with model \eqref{eq:transmodeloutlier} obtains $3\mbox{$\cdot$}32$ non-zero $\gamma_i$'s on average. The masking (M), swamping (S) and joint detection (JD) rates are  $1\mbox{$\cdot$}33\%$, $0\mbox{$\cdot$}77\%$ and $98\%$. Also, our method provides $3\mbox{$\cdot$}22$ non-zero estimated $\beta_j$ on average with the masking, swamping and joint detection rates of  $0\mbox{$\cdot$}67\%$, $4\mbox{$\cdot$}8\%$ and $98\%$ respectively.

For Simulation 4, we choose $\eta_0=1\mbox{$\cdot$}8, p=8, \gamma_{(1:2)}=8, \gamma_3=-8, \gamma_{(4:50)}=0$ and $\beta_0=(3,1\mbox{$\cdot$}5,0,0,2,0,0,0)$. For our Bayesian method with model \eqref{eq:transmodeloutlier}, we have  $3\mbox{$\cdot$}28$ non-zero $\gamma_i$'s  on average with the masking,  swamping and joint detection rates of $0\%$, $0\mbox{$\cdot$}6\%$ and $100\%$. Also, our method provides $3\mbox{$\cdot$}22$ non-zero estimated $\beta_j$ on average with the masking, swamping and joint detection rates of  $0\%$, $0\%$ and $100\%$ respectively.

When we simulate data from model \eqref{eq:errorNI}, the Bayesian method leads to results similar to the results obtained for using simulation model \eqref{eq:transmodeloutlier}. However, only the Bayesian method using \eqref{eq:transmodeloutlier} provides the identification of the observations with errors of large magnitude. In practice, identification of such observations will facilitate further investigations regarding their measurement accuracy, influence on inference and other exploratory diagnostics. All of our Bayesian models provide better results than the penalized quantile regression and LASSO with respect to average number of non-zeros, masking, swamping and joint detection.

%\begin{table}[htp!]
%\begin{center}
%Simulation model with $\eta_0=0\mbox{$\cdot$}5, p=8, \gamma_{(1:2)}=8, \gamma_3=-8, \beta_0=(3,1\mbox{$\cdot$}5,0,0,2,0,0,0)$.
%\scalebox{0.8}{
%\begin{tabular}[ht!]{l|llllll}
%	\hline
%	Method used & $L/L^*-1$& \# of non-zeros &M($\%$) &S($\%$) &JD($\%$)\\ \hline
%	TBS-SG&-0$\cdot$01 & 3$\cdot$22 & 0$\cdot$67 & 4$\cdot$8 & 98\\ \hline
%	Penalized Quantile&	1$\cdot$04 & 5$\cdot$68 & 0 & 53$\cdot$6 &100  \\ \hline
%	LASSO & 1$\cdot$27 & 4$\cdot$4 & 4$\cdot$67 & 30$\cdot$8 & 86\\ \hline
%	TBS-NI-t&-0$\cdot$01 & 3$\cdot$18 & 0 & 3$\cdot$6 &100  		\\ \hline
%	TBS-NI-slash &-0$\cdot$01 & 3$\cdot$16 & 0 & 3$\cdot$2 & 100		\\	\hline
%	TBS-NI-CN &-0$\cdot$01  & 3$\cdot$2 & 0 & 4 & 100 	\\	\hline
%\end{tabular}}
%\end{center}

%\begin{center}
%Simulation model with $\eta_0=1\mbox{$\cdot$}8, p=8, \gamma_{(1:2)}=8, \gamma_3=-8, \beta_0=(3,1\mbox{$\cdot$}5,0,0,2,0,0,0)$.
%\scalebox{0.8}{
%\begin{tabular}[ht!]{l|llllll}
%	\hline
%	Method used & $L/L^*-1$& \# of non-zeros &M($\%$) &S($\%$) &JD($\%$)\\ \hline
%	TBS-SG&-0$\cdot$05 & 3 & 0 & 0 & 100\\ \hline
%	Penalized Quantile&3$\cdot$54 & 5$\cdot$54 & 0 & 50$\cdot$8 &100  \\ \hline
%	LASSO & 5$\cdot$59 & 4$\cdot$84 &0 & 36$\cdot$8&100\\ \hline
%	TBS-NI-t&-0$\cdot$04 & 3 & 0 & 0 & 100		\\ \hline
%	TBS-NI-slash &-0$\cdot$04 & 3 & 0 & 0 & 100		\\	\hline
%	TBS-NI-CN &-0$\cdot$04 & 3 & 0 & 0 & 100 	\\	\hline
%\end{tabular}}
%\caption{Results of Simulation Studies with Outliers for Using Different Methods of Analysis}
%\label{table:soutlier}
%\end{center}
%\end{table}

\section{Analysis of medical expenditure study}
Our motivating study is the Medical Expenditure Panel Survey \citep{cohen2003design, natarajan2008variance}, called the MEPS study in short, where the response variable is each patient's `total health care expenditures in the year 2002'. Previous analyses of of this study \citep{natarajan2008variance} suggest that the variance of the response is a function of the mean (heteroscedasticity). Often in practice, medical cost data are typically highly skewed to the right, because a small percentage of patients may accumulate extremely high costs compared to other patients, and the variance of total cost tends to increase as the mean increases.

In this article, we focus only on one large cluster because every cluster of MEPS study has different sampling design. After removing only a few patients with missing observations, we have $173$ patients and $24$ potential predictors including age, gender, race, disease history, etc. The minimum cost is 0 and the maximum is \$79660, with a mean \$4584 and median \$1342. For the convenience of computation, we standardize the response (cost) and five potential predictors of the patient: age in 2002, highest education degree attained, perceived health status, body mass index (BMI), and ability to overcome illness (OVERCOME). Rest of the potential predictors are binary variables with values 0 and 1. We analyze this study using our proposed Bayesian models and compare the results with the penalized quantile regression method of \citet{koenker2005quantile}.  For Bayesian methods, we use our transform-both-sides model \eqref{eq:transmodel}, the model of \eqref{eq:transmodeloutlier} with sparse large errors (TBSO-SG in short) and the model of \eqref{eq:errorNI} with contaminated normal marginal error. For each method, we compute an observed residual $y_{i0}-x_i^{\mathsmaller T}\hat{\beta}$, where $y_{i0}$ is the observed un-transformed response and $x_i^{\mathsmaller T}\hat{\beta}$ is the estimated median. The Q-Q plots for the residuals are in Figure \ref{fig:loss1}. 

%%%%%%%Fig 2
\begin{figure}[htp!]
\centering
\includegraphics[scale=0.23]{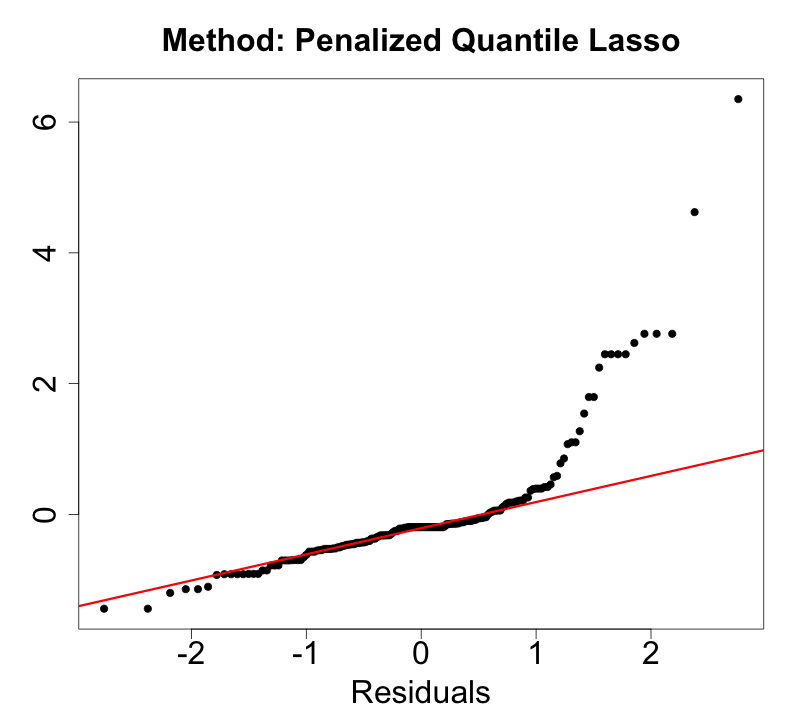}
\includegraphics[scale=0.23]{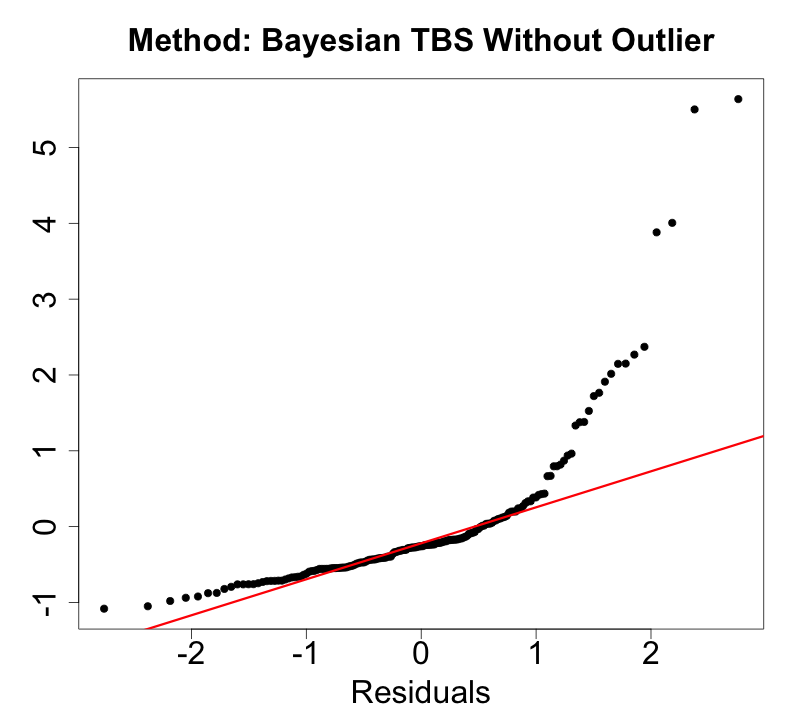}
\includegraphics[scale=0.23]{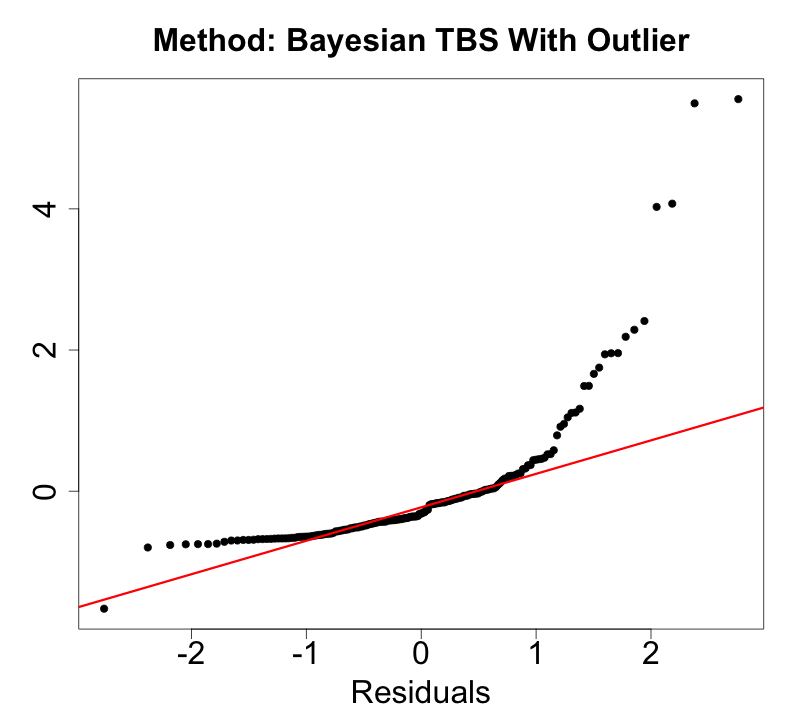}
\includegraphics[scale=0.23]{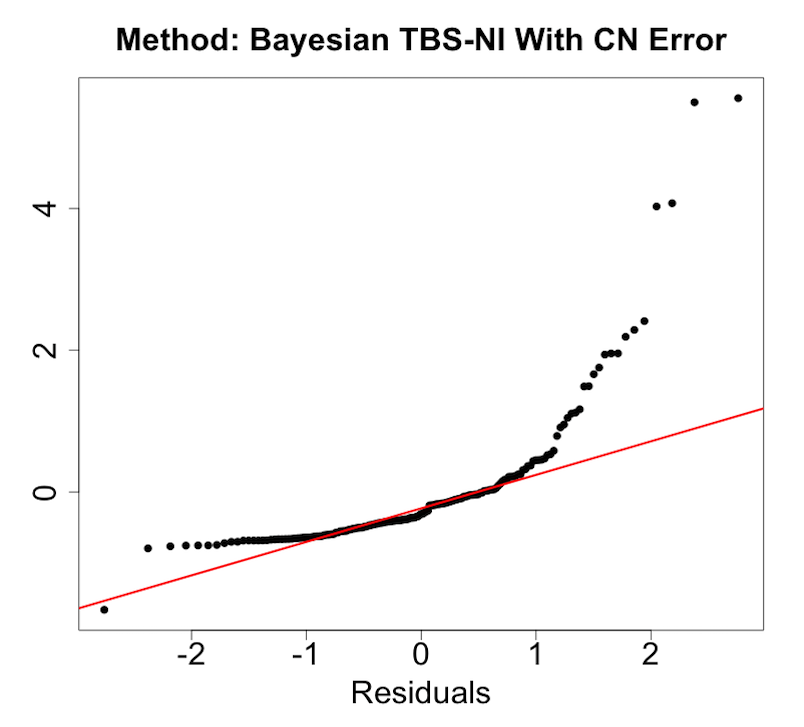}
\caption{Q-Q plots of observed residuals obtained from 4 methods}
\label{fig:loss1}
\end{figure}

From Q-Q plots in Figure \ref{fig:loss1}, it is obvious that the normality assumption 
about un-transformed response is untenable. We now compare the goodness-of-fit of the three Bayesian methods to evaluate their abilities from handling skewness and heteroscedasticity. For this purpose, we use the residual $g_{\hat{\eta}}(y_{i0})-g_{\hat{\eta}}(x_i^{\mathsmaller T}\hat{\beta})$ for the Bayesian TBS-SG model of \eqref{eq:transmodel} and the TBSCN-SG model of \eqref{eq:errorNI}, and the residual $g_{\hat{\eta}}(y_{i0})-g_{\hat{\eta}}(x_i^{\mathsmaller T}\hat{\beta})-\hat{\gamma_i}$ for Bayesian TBSO-SG model of \eqref{eq:transmodeloutlier}, and then display their Q-Q plots in Figure \ref{fig:loss2}.  It is evident from the Q-Q plots that TBSO-SG model of \eqref{eq:transmodeloutlier}  has the best justification to use it for Bayesian analysis, and TBSCN-SG model of \eqref{eq:errorNI} also performs well except may be for some observations in both tails.

%%%%%%Fig 3
\begin{figure}[htp!]
\centering
\includegraphics[scale=0.25]{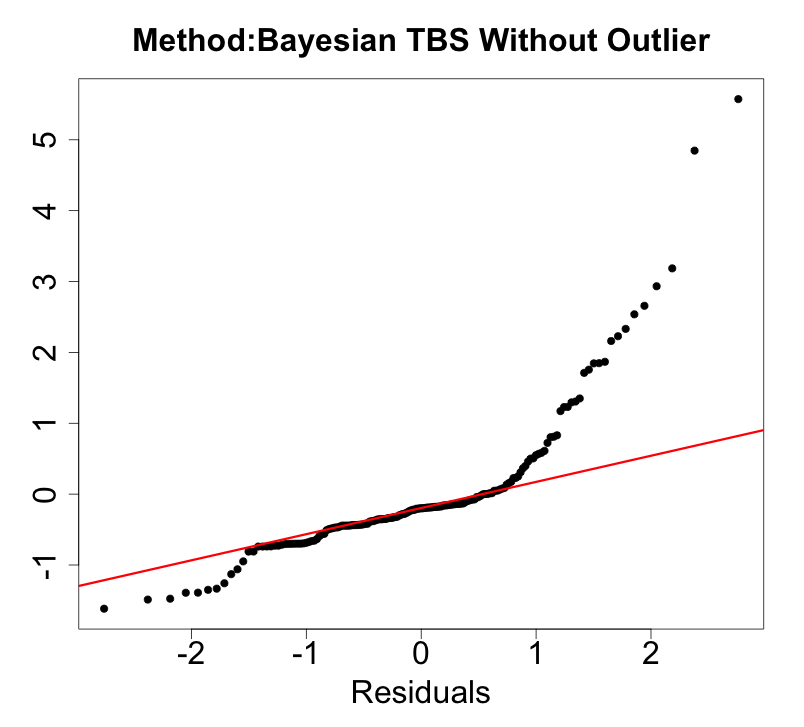}
\includegraphics[scale=0.25]{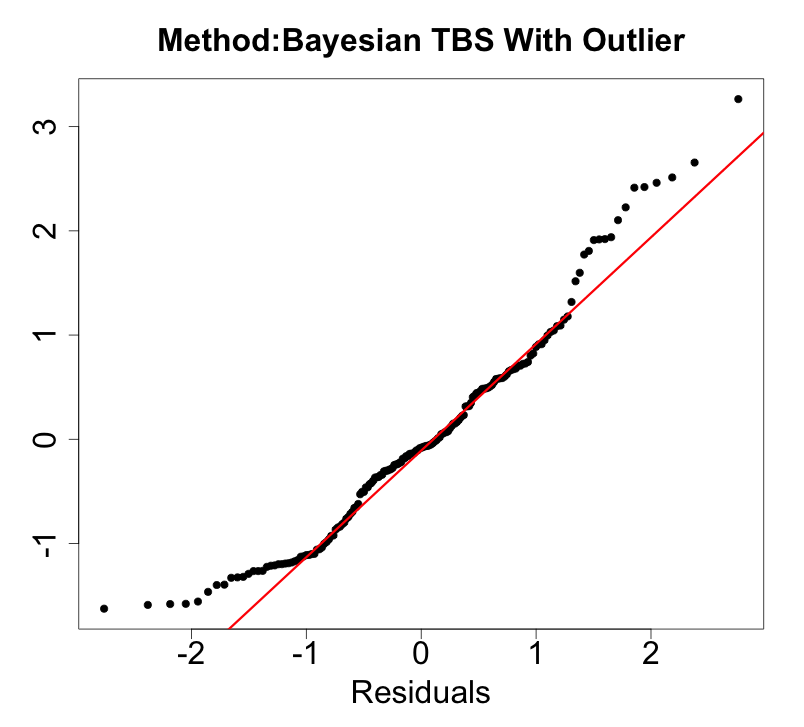}
\includegraphics[scale=0.25]{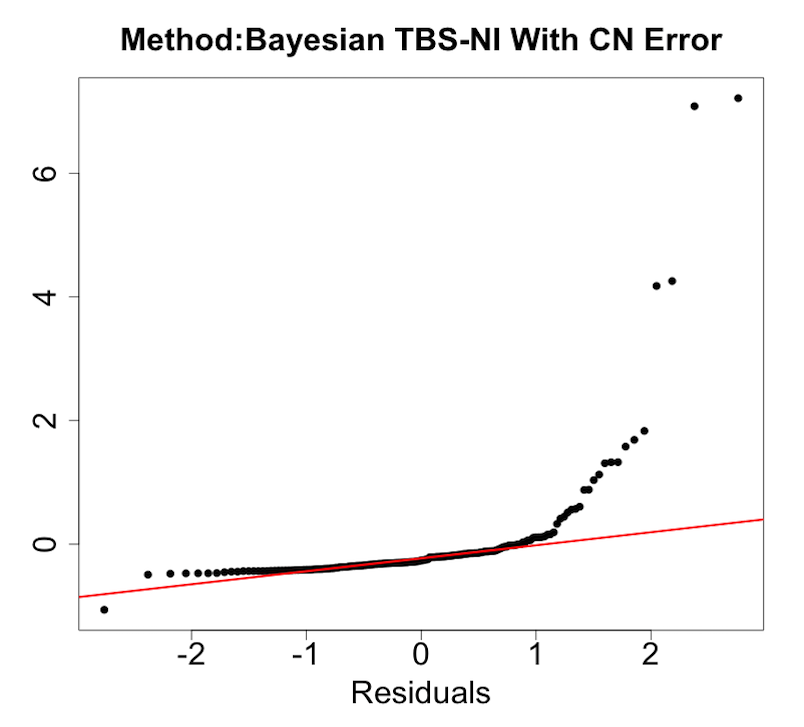}
\caption{Q-Q plots for residuals of transformed responses obtained from 3 Bayesian models}
\label{fig:loss2}
\end{figure}

Using our Bayesian model of \eqref{eq:transmodeloutlier}, we find large posterior evidence of effects of OVERCOME variable  with posterior \text{mean}=$-0\mbox{$\cdot$}17$ and $95\%$ credible interval ($-0\mbox{$\cdot$}21, -0\mbox{$\cdot$}12$), stroke with posterior \text{mean}=$0\mbox{$\cdot$}92$ and $95\%$ credible interval ($0\mbox{$\cdot$}66, 1\mbox{$\cdot$}23$), and medication with posterior \text{mean}=$-0\mbox{$\cdot$}35$ and $95\%$ credible interval ($-0\mbox{$\cdot$}41, -0\mbox{$\cdot$}24$). Model\eqref{eq:errorNI} identifies these same predictors of model \eqref{eq:transmodeloutlier} with slightly different interval estimates. Model \eqref{eq:transmodel} also identifies three predictors with large posterior evidence 
of effects: perceived health status, stroke and the indicator of major ethnic group. Stroke is the only variable with large posterior evidence 
of effects in all three models. Even though penalized quantile regression based analysis selects a larger number of predictors compared to the number of predictors selected by our models, the only statistically significant variable from quantile regression analysis is age (estimate of $0\mbox{$\cdot$}08$ with standard error $0\mbox{$\cdot$}03$). This may be explained by the larger estimated standard errors of the estimates from quantile regression compared to the posterior standard deviations of the corresponding parameters obtained via Bayesian analysis.

In order to better understand the prediction performance on observed data, we present a scatter plot with overlaid quantile lines for each Bayesian method in Figure \ref{fig:scatter}. For each method, we display scaled $x_i^{\mathsmaller T}\hat{\beta}$ and scaled observed response $y_{i0}$, along with estimated 25th percentile and 75th percentile curves using $g_{\hat{\eta}}^{-1}\{g_{\hat{\eta}}(x_i^{\mathsmaller T}\beta)+Z_{\alpha}^*\}$, where $Z_{\alpha}^*$ is estimated $\alpha$-percentile of $f_e(\cdot)$. Figure \ref{fig:scatter} shows that the method using \eqref{eq:transmodeloutlier} explains the observed data better than methods using \eqref{eq:transmodel} and \eqref{eq:errorNI}. The observations with large errors identified by analysis using \eqref{eq:transmodeloutlier} are marked by asterisk signs in the second plot. We find that all the observations identified by \eqref{eq:transmodeloutlier} are outside the estimated interquartile ranges. It shows that our transform-both-sides model of \eqref{eq:transmodeloutlier} is successful in handling data with skewness, heteroscadesticity as well as very large errors in few subjects. Model of \eqref{eq:errorNI} is also able to handle skewness and heteroscadesticity but is not able to identify observations with extremely large errors.

%%%%%%%Fig 4
\begin{figure}[htp!]
\centering
\includegraphics[scale=0.25]{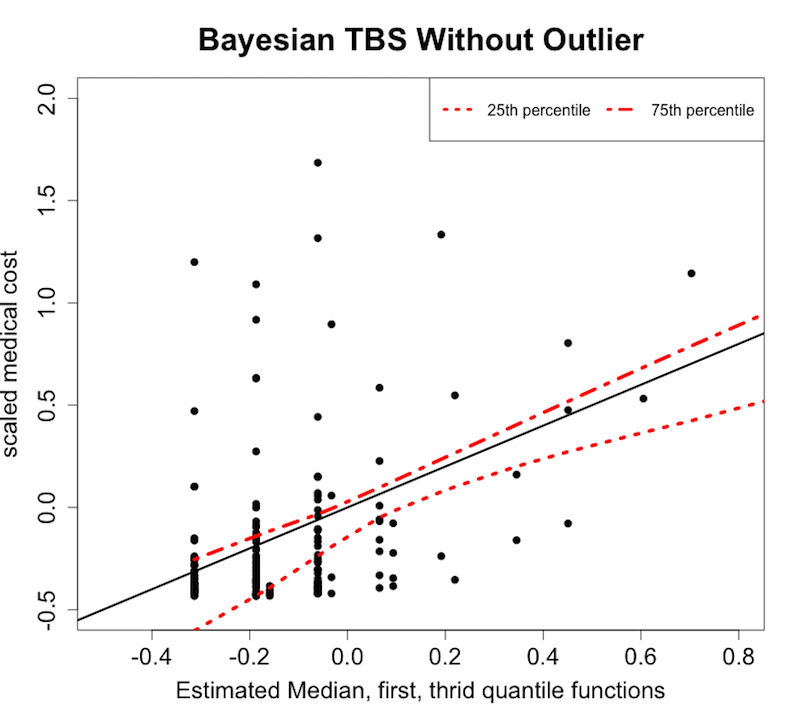}
\includegraphics[scale=0.25]{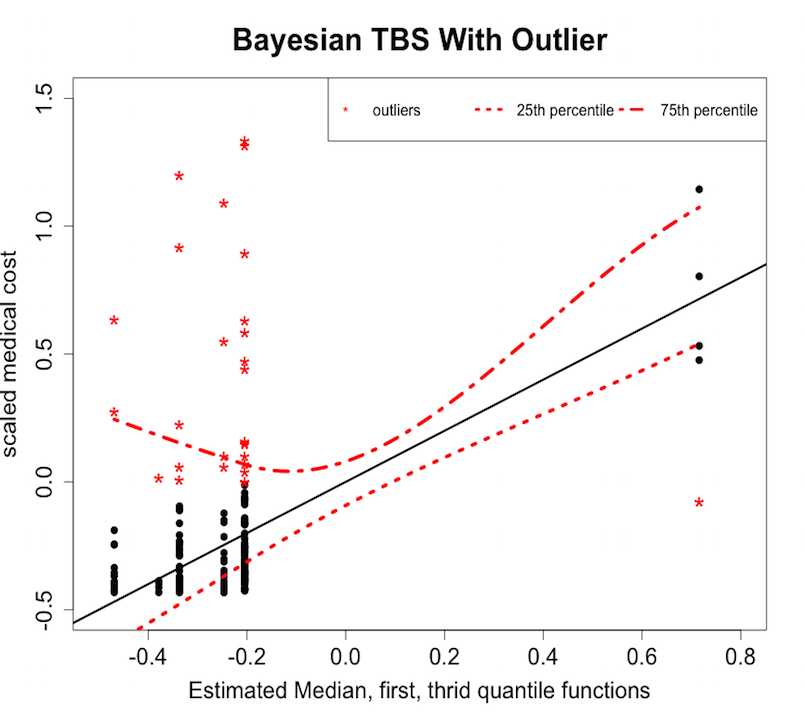}
\includegraphics[scale=0.25]{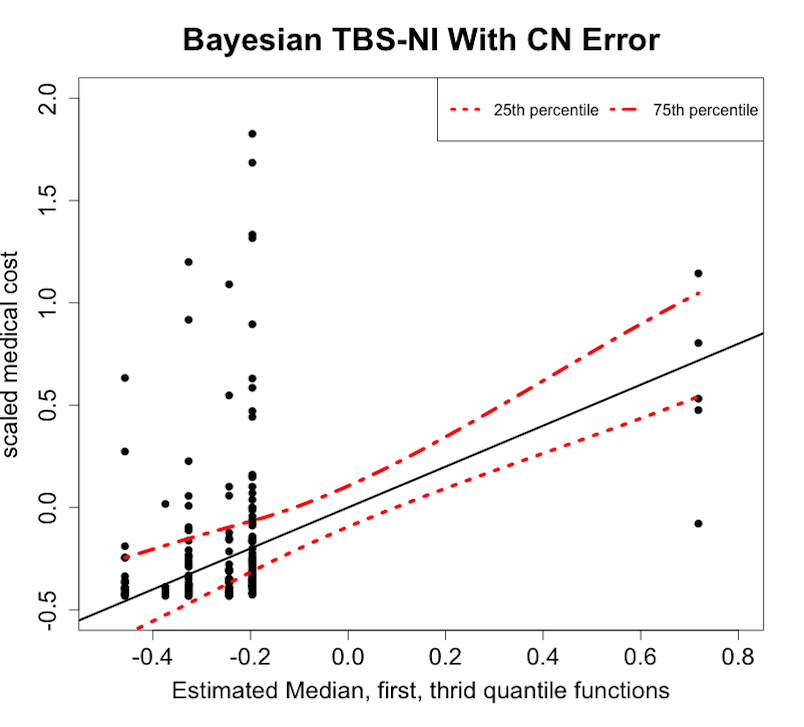}
\caption{Scatter plots of scaled observed responses and quantile regression functions obtained from 3 Bayesian models}
\label{fig:scatter}
\end{figure}

We also use posterior predictive loss approach \citep{GELFAND01031998} to evaluate the prediction accuracy under each Bayesian method. We compute the prediction errors of our Bayesian methods by $\sum_{i=1}^n\mbox{E}[\{g_{\hat{\eta}}(y_{i0})-g_{\hat{\eta}}(x_i^{\mathsmaller T}\hat{\beta})\}^2|D]$ from model \eqref{eq:transmodel}, and by $\sum_{i=1}^nE[\{g_{\hat{\eta}}(y_{i0})-g_{\hat{\eta}}(x_i^{\mathsmaller T}\hat{\beta})-\gamma_i\}^2|D]$ for model \eqref{eq:transmodeloutlier} using MCMC approximation, where $D$ is the observed dataset. The average prediction error from model \eqref{eq:transmodeloutlier} is $15\cdot03$, which is considerably better than model \eqref{eq:transmodel} with average prediction error $174\cdot45$.

\section{Discussion}
In this article, we propose Bayesian  variable selection methods for skewed and heteroscedastic response. The methods are highly suitable for modeling, computation, analysis and interpretation of real-life health care cost studies, where we aim to determine and estimate effects of a sparse set of explanatory variables for health care expenditures out of a large set of potential explanatory variables. Simulation results indicate a better performance of our Bayesian methods compared to existing frequentist quantile regression tools. Also, our Bayesian approaches provide flexible and robust estimations to incorporate a wide variety of practical situations. The advantages of our Bayesian methods include their practical and easy implementation using standard statistical software. In the appendix, we prove the consistency of variable selection even when the number of potential predictors $p$ is comparable to, however, smaller than  $n$. The proofs are only provided for a special case of the covariate matrix and when the transform parameter $\eta$ is known. Proof for a more general case can be obtained following a similar, but more tedious mathematical arguments.
% \end{spacing}

\bibliographystyle{agsm}
\bibliography{paper-ref}

\newpage

\appendix
%\appendixpage
%\addappheadtotoc
\section*{Proof of Theorem \ref{thm:1} }\label{ssec:A1}
For $\gamma \in \{0, 1\}^p$ and $x \in \mathbb{R}^p$, define $x_\gamma \in \mathbb{R}^{p_{\gamma}}$ to be the vector 
$(x_j: \gamma_j \neq 0)$.
The log-likelihood corresponding to \eqref{eq:transmodel}  has the expression
\begin{equation}\label{eq:logl}
\ell(\beta)
= \sum_{i=1}^n\log\big(g_{\eta}'(y_i)\big)-\frac{n}{2}\log(2\pi)-n\log(\sigma)-\frac{1}{2\sigma^2}\norm{g_\eta(y)-g_\eta(X\beta)}^2,
\end{equation}
%\begin{eqnarray}\label{eq:logl}
%\ell(\beta, \sigma)&=&\sum_{i=1}^n\log\big(g_{\eta}'(y_i)\big)-\frac{n}{2}\log(2\pi)-n\log(\sigma)-\frac{1}{2\sigma^2}\sum_{i=1}^n\big(g_\eta(y_i)-g_\eta(x_i^T\beta)\big)^2\\ \nonumber
%&=& \sum_{i=1}^n\log\big(g_{\eta}'(y_i)\big)-\frac{n}{2}\log(2\pi)-n\log(\sigma)-\frac{1}{2\sigma^2}\big(g_\eta(y)-g_\eta(X\beta)\big)^T(g_\eta(y)-g_\eta(X\beta)\big)
%\end{eqnarray}
where $g_\eta(y) =\big(g_\eta(y_1), \cdots, g_\eta(y_n) \big)^\T$, $g_\eta(X\beta)=\big\{g_\eta(x_1^\T\beta), \cdots, g_\eta(x_n^\T\beta)\big\}^\T$.  The gradient of \eqref{eq:logl} is given by
\begin{equation}\label{eq:grad}
\nabla\ell(\beta)=\frac{1}{\sigma^2}X^\T\bar{D}_\eta\big(g_\eta(y)-g_\eta(X\beta)\big),
\end{equation}
where $\bar{d}_\eta=(g_\eta'(x_1^T\beta), \cdots, g_\eta'(x_n^T\beta))$ with $g_\eta'(x_i^T\beta)=\abs{x_i^T\beta}^{\eta-1}$ and $\bar{D}_\eta=\text{diag}(\bar{d}_\eta)$.  $\beta \mapsto \nabla\ell(\beta)$ is continuous. Since the number of true signals $\beta_j^*$ is finite, we  assume $\beta_j^* \in [M_L, M_U]$. Hence, as long as $\beta$ is in a neighborhood of $\beta^*$,   $\nabla\ell(\beta)$ is bounded 
%with respect to $\beta$ 
%(\textcolor{red}{This is not correct!}) 
since $\abs{x_i^T\beta}^{\eta-1}$ in \eqref{eq:grad} is bounded when $\eta \in (1, 2)$.  The Hessian of \eqref{eq:logl} is defined as

\begin{equation}\label{eq:Hessian}
H(\beta)=\frac{1}{\sigma^2}X^T\big(-\bar{D}_\eta^T\bar{D}_\eta+\bar{\bar{D}}_\eta^T M_\beta(y, X) \big)X,
\end{equation}
where $\bar{\bar{d}}_\eta$ is defined as the element-wise second derivative of $g_\eta(x_i^T\beta)$ on $\beta$ 
with $g_\eta''(x_i^T\beta)=\text{\mbox{sgn}}(x_i^T\beta)(\eta-1)\abs{x_i^T\beta}^{\eta-2}$ and $\bar{\bar{D}}_\eta=\text{diag}(\bar{\bar{d}}_\eta)$. Meanwhile, $m_\eta(y, X)= \big(g_\eta(y_1)-g_\eta(x_1^T\beta), \cdots, g_\eta(y_n)-g_\eta(x_n^T\beta)\big)$ and $M_\eta(y, X) = \text{diag}(m_\eta(y, X))$.   Using Laplace approximation, the Bayes factor can be approximated as 
\begin{equation}\label{eq:bayesfac}
\frac{p(y|\beta, \gamma)}{p(y|\beta^*, \gamma^*)} =(2\pi)^{\frac{p_\gamma-p_{\gamma^*}}{2}} \times e^{\{\ell(\tilde{\beta}_{\gamma})-\ell(\beta^*_{\gamma^*})\}} \times \frac{\pi(\tilde{\beta}_{\gamma}|\gamma)}{\pi(\beta^*_{\gamma^*}|\gamma^*)} \times \frac{|H(\beta^*_{\gamma^*}|\gamma^*)|^{1/2}}{|H(\tilde{\beta}_\gamma|\gamma)|^{1/2}}. 
\end{equation}
The first term in the r.h.s of \eqref{eq:bayesfac} is $\big(O_p(1)\big)^n$ with $p \leq n$. For the second term, we denote by $\lambda(y)=\ell(\tilde{\beta})-\ell(\beta^*)$ the likelihood ratio statistic. As in \citet[Proposition~3]{rossell2015non}, it is straightforward to verify that our sampling model \eqref{eq:transmodel} satisfies Walker's conditions (A1)-(A5) and (B1)-(B4) \citep{Walker1969}. Hence, our MLE $\tilde{\beta_i}$ is consistent and the Hessian matrix $H(\tilde{\beta}_\gamma|\gamma)$ in \eqref{eq:bayesfac} converges in probability. We consider two cases below.
%Suppose that the sampling model \eqref{eq:transmodel} satisfies the conditions in \citet{Walker1969} \textcolor{red}{[Why do we need this? If yes, verify walker conditions are true or at least say that the conditions are satisfied if trivial]}, we consider two cases below. 

When $\gamma^* \not\subset \gamma$, i.e., $\gamma$ misses some true active coefficients, the second term $e^{(\text{log}L(\tilde{\beta}_{\gamma})-\text{log}L(\beta^*_{\gamma^*}))} \xrightarrow{P} e^{-n\text{KL}(p(y|\beta^*, \gamma^*), p(y|\beta, \gamma))}$
% (\textcolor{red}{Be careful about the notation here: $p(y|\gamma^*)$ actually depend on $\beta$, so need to use $p(y|\beta^*, \gamma^*)$})
  where $\text{KL}(p(y|\beta^*, \gamma^*), p(y|\beta, \gamma))$ is the Kullback-Leibler divergence between optimal $p(y|\beta^*, \gamma^*)$ and $p(y|\beta, \gamma)$ under $\gamma$. Here the minimum KL divergence $\text{KL}(p(y|\beta^*, \gamma^*), p(y|\beta, \gamma))$ is strictly positive since, 
\begin{eqnarray}\label{eq:KLpostive} \nonumber
\text{KL}(p(y|\beta^*, \gamma^*), p(y|\beta, \gamma)) &=& \sum_{i=1}^n \int [\text{log}(p(y_i|\beta^*, \gamma^*))-\text{log}(p(y_i|\beta, \gamma))]p(y_i|\beta^*, \gamma^*)d_{y_i}\\ \nonumber
&=& \frac{1}{2\sigma^2}\sum_{i=1}^n \Big( E \big( g_\eta(y_i)-g_\eta(x_i^T\beta) \big)^2-E\big( g_\eta(y_i)-g_\eta(x_i^T\beta^*)\big) ^2 \Big)\\
&=& \frac{1}{2\sigma^2} \sum_{i=1}^n \big( g_\eta(x_i^T\beta^*)-g_\eta(x_i^T\beta) \big)^2 > 0
\end{eqnarray}
when $\beta \neq \beta^*$ satisfies the eigenvalue conditions in \S \ref{sec:cons} that indicates no linear dependency among covariates $x_i$. Therefore the second term is $O_p(e^{-n})$ when $\gamma^* \not\subset \gamma$.
%\textcolor{red}{My main point of concern is that the $\gamma$ (change the notation for dual use) here is not an universal constant, so this argument is flawed. You'd need to look at Rossell and Tolesca's supplement to  make this rigorous. I think that we have to look at  minimum Kullback Leibler distance 
%$p(y | \beta^*, \gamma^*)$ and $p(y | \beta, \gamma)$ over all $\beta$ and  $\gamma \not\subset \gamma^*$. The fact that this is positive  is a much more plausible  assumption. } 

%we borrowed the idea from \citet{johnson2012nonlocal} which split the first term according to Boole's inequatlity. Denote $u^*=\gamma \cup \gamma^*$, The first term is upper bounded by $O(\text{exp}(\text{log}L(\tilde{\beta_{\gamma}})-\text{log}L(\beta_{u^*})))+O(\text{exp}(-(\text{log}L(\beta_{\gamma^*})-\text{log}L(\beta_{u^*}))))$ which is dominated by $O_p(e^{-cn})$ with $c>0$ and $n \rightarrow \infty$. In the extreme case, $u^*=\gamma^*$ and $\text{log}L(\tilde{\beta_{\gamma}})-\text{log}L(\beta_{\gamma^*}))=O_p(n)$, so the first term of equation \eqref{eq:bayesfac} decreased exponentially fast with n.
%$log(p(y|\gamma)/p(y|\gamma^*))=O_p(n)$ 

When $\gamma^* \subset \gamma$, we denote the likelihood-ratio statistic  by $\Lambda(y)=\text{log}L(\tilde{\beta}_{\gamma})-\text{log}L(\beta^*_{\gamma^*})$. Under appropriate regularity conditions \citep{hogg2013introduction}, our likelihood ratio statistic $\Lambda(y)$ is asymptotically chi-square distributed. The regularity conditions relevant to the argument are listed as (R0)-(R9) in \citet{hogg2013introduction} where (R0)-(R2) and (R6)-(R8) can be obtained trivially. Conditions (R3) and (R4) are related to Fisher information and are satisfied by Hessian matrix \eqref{eq:Hessian}. Conditions (R5) and (R9) essentially guarantee that the remainder of a second order Taylor expansion around $\beta$ is bounded in probability. To that end, note that
\begin{eqnarray}\nonumber
& & \left \vert \frac{\partial^3}{\partial \beta_j\partial \beta_k\partial \beta_l} \ell(\beta) \right\vert \\ \nonumber
&=&\left \vert \frac{\eta-1}{\sigma^2}\sum_{i=1}^nx_{ij}x_{ik}x_{il} \Big[ -\big(2+\mbox{sgn}(x_i^T\beta)\big)|x_i^T\beta|^{2\eta-3}+(\eta-2)\big(g_{\eta}(y_i)-g_{\eta}(x_i^T\beta)\big)|x_i^T\beta|^{\eta-3}\Big]  \right\vert \\
&\leq& \frac{\eta-1}{\sigma^2}\sum_{i=1}^nx_{ij}x_{ik}x_{il} \Big[ 3|x_i^T\beta|^{2\eta-3}+\mid g_{\eta}(y_i)-g_{\eta}(x_i^T\beta) \mid |x_i^T\beta|^{\eta-3}\Big] := M_{jkl}(y; X)
\end{eqnarray} 
%Our model \eqref{eq:transmodel} satisfies the regularity conditions of \citet{huber1973} that $\nabla \text{log}L$ is be bounded and continuous, as well as the errors $\epsilon_i$ are independent and identically distributed. %\textcolor{red}{These regularity conditions should be stated clearly and verified. These are Walker's conditions as referred in Rossell and Telesca paper } 
where $E[M_{jkl}(y;x)] < \infty$ for all $j, k, l \in 1, \cdots, p$. Therefore our model \eqref{eq:transmodel} satisfies all regularity conditions implying $\Lambda(y) \sim \chi_{p_{\gamma}-p_{\gamma^*}}^2$ and hence $O_p(1)$
%$\big(O_p(1)\big)^{p_{\gamma}-p_{\gamma^*}}$ 
as required. When $\gamma$ has moderate size with $p_{\gamma} \leq n$, we will show the first term in \eqref{eq:bayesfac} is dominated by the second term later.

Next consider the second term $\pi(\tilde{\beta}|\gamma)/\pi(\beta^*|\gamma^*)$ under a non-local prior with $\eta \in (1, 2)$ is known.
\begin{eqnarray}\label{eq:priorfac}  \nonumber
\frac{\pi(\tilde{\beta}|\gamma)}{\pi(\beta^*|\gamma^*)}&=&\frac{\prod_{i=1}^{p_\gamma}\frac{1}{\sqrt{2\pi} \sigma_\beta}\text{exp}(-\tilde{\beta}_{\gamma _i}^2/2\sigma_\beta^2)|\tilde{\beta}_{\gamma_i }|^{\eta-1}}{\prod_{j=1}^{p_{\gamma^*}}\frac{1}{\sqrt{2\pi} \sigma_\beta}\text{exp}(-\beta^{*2}_{\gamma^*_j}/2\sigma_\beta^2)|\beta^*_{\gamma^*_j}|^{\eta-1}}\\
&=&(\frac{1}{\sqrt{2\pi}\sigma_\beta})^{p_\gamma-p_{\gamma^*}}\text{exp}\big(-(\sum \tilde{\beta}_{\gamma_i}^2-\sum \beta^{*2}_{\gamma^*_j})/2\sigma_\beta^2 \big)\frac{\prod|\tilde{\beta}_{\gamma_i}|^{\eta-1}}{\prod|\beta^*_{\gamma^*_j}|^{\eta-1}}
\end{eqnarray}

First if $\gamma^* \subset \gamma$, given that $\tilde{\beta}_{\gamma_i}=O_p(n^{-1/2})$ and $\eta-1>0$, the second term is ensentially $O_p(n^{-(\eta-1)(p_\gamma-p_{\gamma^*})/2})$. When $\gamma^* \not\subset \gamma$, denote $s^*=\gamma \cap \gamma^*$, the second term is $O_p(n^{-(\eta-1)(p_\gamma-p_{s^*})/2})$ and upper bounded by $O_p(n^{-(\eta-1)/2})$. 
%\textcolor{red}{Johnson and Rossell needs $r \geq 2$ for the consistency proof to go through. Why do not you need similar restrictions on $\eta$ for your proof? }

To conclude the proof, we need to deal with the the third term in \eqref{eq:bayesfac}. The Hessian matrix $H(\tilde{\beta})$ is given by \eqref{eq:Hessian}
%$\frac{1}{\sigma^2}X^T\big(-\bar{D}_\eta^T\bar{D}_\eta+\bar{\bar{D}}_\eta^T M_\beta(y, X) \big)X$ 
with each element $H_{ij}=X_i^TX_jO_p(1)$. Since $\tilde{\beta_\gamma}$ converge in probability to $\beta_{\gamma^*}$, we have $n^{-1}H(\tilde{\beta}|\gamma) \xrightarrow{p} H(\tilde{\beta}|\gamma)$. Therefore by continous mapping theorem, 
%$X^TX$ is positive definite almost surely as $n \rightarrow \infty$ and defined by continuous function of eigenvalues $\prod_{i=1}^p \lambda_i$ Then
 the third term in \eqref{eq:bayesfac} is approximated as 
%the element $h_{ij}$ in $n^{-1}H(\tilde{\beta}|\gamma)$ is given by the $(i, j)$ element in $\frac{X^TX}{n}$ which is $O_p(1)$. Then the third term in \eqref{eq:bayesfac}

%  Assuming $c_1<\frac{X_i^TX_j}{n}<c_2$ thus $H(\tilde{\beta})$ is $O_p(n)$. Therefore,

\begin{equation}
 \frac{\abs{H(\beta^*|\gamma^*)}^{1/2}}{\abs{H(\tilde{\beta}|\gamma)}^{1/2}} \asymp \frac{n^{p_{\gamma^*/2}}}{n^{p_\gamma/2}}\frac{\abs{n^{-1}H(\beta_0|\gamma^*)}}{\abs{n^{-1}H(\tilde{\beta}|\gamma)}}  = O_p(n^{(p_{\gamma^*}-p_\gamma)/2}). 
\end{equation}

 To conclude, the Bayes factor \eqref{eq:bayesfac} is $O_p(n^{-\eta(p_{\gamma}-p_{\gamma^*})/2})$ when $\gamma^* \subset \gamma$,  and  $O_p(n^{(p_{\gamma^*}-p_{\gamma})/2}e^{-n})$ when $\gamma^* \not \subset \gamma$. Then the posterior probability $p(\gamma=\gamma^* \mid y)$ can be lower-bounded as
 \begin{eqnarray} \nonumber
 p(\gamma=\gamma^* \mid y) &=& \frac{p(y \mid \gamma^*)\pi(\gamma^*)}{p(y \mid \gamma^*)\pi(\gamma^*)+\sum_{\gamma \neq \gamma^*}p(y \mid \gamma)\pi(\gamma)} \\ \nonumber
 &=& \Big[ 1+\sum_{\gamma^* \subset \gamma}\frac{p(y \mid \gamma)\pi(\gamma)}{p(y \mid \gamma^*)\pi(\gamma^*)} + \sum_{\gamma^* \not \subset \gamma}\frac{p(y \mid \gamma)\pi(\gamma)}{p(y \mid \gamma^*)\pi(\gamma^*)} \Big]^{-1}\\ \nonumber
 &\geq& \Big[ 1+ \sum_{\gamma^* \subset \gamma}O_p(n^{-\eta(p_{\gamma}-p_{\gamma^*})/2}) + \sum_{\gamma^* \not \subset \gamma}O_p(n^{(p_{\gamma^*}-p_{\gamma})/2}e^{-n}) \Big]^{-1}\\ \nonumber
 &\geq& \Big[ 1+ \sum_{\gamma^* \subset \gamma}O_p(n^{-(p_{\gamma}-p_{\gamma^*})/2}) + \sum_{\gamma^* \not \subset \gamma}O_p(e^{-n}) \Big]^{-1} \xrightarrow{a.s.} 1   
 \end{eqnarray}
 which concludes the proof.  
 %\textcolor{red}{Please explain clearly how got the last inequality. You have to count how many 
 %$\gamma^* \not \subset \gamma$ and $\gamma^*  \subset \gamma$. Also state clearly what happens to the ratio $\pi(\gamma)/\pi(\gamma^*)$}.  

\end{document}